\documentclass[10pt,compsocconf]{IEEEtran}
\ifCLASSINFOpdf
   \usepackage[pdftex]{graphicx}
   \graphicspath{{../pdf/}{../jpeg/}}
   \graphicspath{./figcoper/}
   \DeclareGraphicsExtensions{.pdf,.jpeg,.png}
\else
  \usepackage[dvips]{graphicx}
   \graphicspath{{../eps/}}
  \graphicspath{./figcoper/}
  \DeclareGraphicsExtensions{.eps}
\fi
\usepackage{graphicx}
\graphicspath{{./figcoper/}}

\usepackage[cmex10]{amsmath}
\usepackage[tight,footnotesize]{subfigure}
\usepackage{url}

\usepackage{comment}
\usepackage{algpseudocode}
\usepackage{algorithm}
\usepackage{multirow}
\usepackage{lineno}
\graphicspath{./figcoper/}
\usepackage{amsthm}

\usepackage{amssymb}
\newtheorem{theorem}{Theorem}
\newtheorem{lemma}{Lemma}
\newtheorem{corollary}{Corollary}
\newtheorem{fact}{Fact}
\date{xx}

\newcommand{\I}{\mathcal{I}}
\newcommand{\bfx}{\mathbf{x}}

\newcommand{\tr}{\operatorname{tr}}
\newcommand{\rank}{\mathrm{rank}}

\begin{document}
%
\title{Information Splitting for Big Data Analytics  }


\author{\IEEEauthorblockN{Shengxin Zhu, Tongxiang Gu, Xiaowen Xu and Zeyao Mo}\\
\IEEEauthorblockA{Laboratory of Computational Physics,\\ Institute of Applied Physics and Computational Mathematics,\\
P.O.Box 8009, Beijing 100088, P.R. China\\
Email: {\{zhu\_shengxin;txgu;xwxu;zeyao\_mo\}}@iapcm.ac.cn}
}


%


\maketitle

\begin{abstract}
Many statistical models require an estimation of unknown (co)-variance parameter(s) in a model. The estimation usually obtained by maximizing a log-likelihood which involves log determinant terms. In principle, one requires the \emph{observed information}--the negative Hessian matrix or the second derivative of the log-likelihood---to obtain an accurate maximum likelihood estimator according to the Newton method. When one uses the \emph{Fisher information}, the expect value of the observed information, a simpler algorithm than the Newton method is obtained as the Fisher scoring algorithm. With the advance in high-throughput technologies in the biological sciences, recommendation systems and social networks, the sizes of data sets---and the corresponding statistical models---have suddenly increased by several orders of magnitude. Neither the observed information nor the Fisher information is easy to obtained for these big data sets. This paper introduces an information splitting technique to simplify the computation. After splitting the mean of the observed information and the Fisher information, an simpler approximate Hessian matrix for the log-likelihood can be obtained. This approximated Hessian matrix can significantly reduce computations, and makes the linear mixed model applicable for big data sets. Such a spitting and simpler formulas heavily depends on matrix algebra transforms, and applicable to large scale breeding model, genetics wide association analysis.

\end{abstract}


 \ifCLASSOPTIONpeerreview
 \begin{center} \bfseries EDICS Category: 3-BBND \end{center}
\fi
%
\IEEEpeerreviewmaketitle

\section{Introduction}

Many applications in animal/plant breeding \cite{masud13}, clinic trials, ecology and evolution \cite{Cell}, genome-wide association \cite{lip11,lis12,zhang10,zhou12} involve the following linear mixed model.
\begin{equation}
y=X\tau+Zu+\epsilon, \label{eq:LMM}
\end{equation}
where $y\in \mathbb{R}^{n\times 1}$ is a vector which consists of $n$ observations, $\tau\in
\mathbb{R}^{p\times 1}$ is a vector of $p$ fixed effects, $X \in
\mathbb{R}^{n\times p}$ is the \emph{design matrix} which corresponds to
fixed effects, $u \in \mathbb{R}^{b\times 1}$ is the vector of
unobserved random effects, $Z \in \mathbb{R}^{n\times b}$ is the
design matrix which corresponds to the random effects. $\epsilon \in
\mathbb{R}^{n\times 1}$ is the vector of residual errors. The random effects, $u$, and the residual errors, $\epsilon$,
are multivariate normal distributions such that $E(u)=0$,
$E(\epsilon)=0$, $u\sim N(0, \sigma^2 G)$, $\epsilon \sim N(0,
\sigma^2 R)$ and
\begin{equation}
\text{var}\left[\begin{array}{c}
u\\
\epsilon
\end{array}\right]=\sigma^{2}\left[\begin{array}{cc}
G & 0\\
0 & R
\end{array}\right],
\end{equation}
where $G\in \mathbb{R}^{b\times b }$, $R \in \mathbb{R}^{n\times
n}$.

When the co-variance matrices $G$ and $R$ are known, one can obtain the \emph{Best Linear Unbiased Estimators} (BLUEs), $\hat{\tau}$, for the fixed effects and the \emph{Best Linear Unbiased Prediction} (BLUP), $\tilde{u}$, for the random effects according to the maximum likelihood method, the Gauss-Markov-Aitiken least square \cite[\S 4.2]{Rao}. $\hat{\tau}$ and $\tilde{u}$ satisfy the following equations
\begin{equation}
C\begin{pmatrix}
\hat{\tau}  \\ \tilde{u}
\end{pmatrix}
=\begin{pmatrix}
X^TR^{-1}y \\ Z^T R^{-1}y
\end{pmatrix},
\label{eq:mme}
\end{equation}
where
$$
C=\begin{pmatrix}
X^TR^{-1}X  & X^T R^{-1}Z \\
Z^TR^{-1}X  & Z^TR^{-1}Z+G^{-1}
\end{pmatrix}.
$$
For such a forward problem, confidence or uncertainty of the estimations of the fixed and random effects can be quantified in term of co-variance of the estimators and the predictors
\begin{equation}
\mathrm{var} \begin{pmatrix}
\hat{\tau} \\
\tilde{u} -u
\end{pmatrix}
=\sigma^2 C^{-1}.
\label{eq:pre}
\end{equation}

In many other more interesting and more realistic cases, the co-variance matrices of the random effects and the random noise are unknown.
Typically $G$ and $R$ are parametric matrices with parameters that need to
be estimated, say, $G=G(\gamma)$ and $R=R(\phi)$. We shall denote $\kappa=(\gamma, \phi)$ and $\theta :=(\sigma^2, \kappa)$. To obtain an estimation on the confidence or uncertainty as in \eqref{eq:pre} one has to obtain an estimate on the variance parameter $\theta$ first.
Efficient statistical estimates for the variance parameters via the \emph{REstricted Maximum Likelihood} (REML) method 
requires to maximize the corresponding (restricted) log-likelihood function
\cite[p.252]{Searle06}(details follows in next section): 
\begin{multline}
\ell_R(\theta)=\mathrm{const}-\frac{1}{2}\{ (n-\nu)\log\sigma^2 +\log
\det(H)  \\ +\log \det(X^TH^{-1}X) +\frac{y^TPy}{\sigma^2} \},
\label{eq:ellR}
\end{multline}
where $\nu=\rank(X)$,  $$H=R(\phi)+ZG(\gamma)Z^T,$$  and
\begin{equation}
P=H^{-1}-H^{-1}X(X^TH^{-1}X)^{-1}XH^{-1}. \label{eq:p}
\end{equation}
Here we suppose $\nu=p$.  To find an estimate for the variance parameter, one has to maximum of $\ell_{R}(\theta)$ according to the maximum likelihood principle.
This is a conceptually simple nonlinear Newton-Raphson iterative procedure as described in Algorithm~\ref{alg:NR}, however it can be computationally expensive to carry out for big data sets. 

Usually, one obtain the maximum or an quasi-maximum value by finding the \emph{stationary point} of the restricted log-likelihood, the zeros
of its first derivatives with respects to $\theta$.
These first derivatives of $\ell_R$ are referred to as the
\textit{scores} for the variance parameters $\theta$ or the log-likelihood \cite[p.252]{Searle06}:
\begin{align}
s(\sigma^2)&=\frac{\partial \ell_R}{\partial \sigma^2}
=-\frac{1}{2}\left\{
\frac{n-p}{\sigma^2}-\frac{y^TPy}{\sigma^4}\right \}, \label{eq:ssigma}\\
s(\kappa_i)&=\frac{\partial \ell_R}{\partial
\kappa_i}=-\frac{1}{2}\left\{ \tr(P\frac{\partial H}{\partial
\kappa_i}) -\frac{1}{\sigma^2}y^TP \frac{\partial H}{\partial
\kappa_i}P y  \right\}.
\end{align}
We shall denote $$S(\theta)=(s(\sigma^2),s(\kappa_1), \ldots,
s(\kappa_m))^T $$ as the \emph{score vector}. The negative of the Hessian of the log-likelihood
function, or the negative Jacobian matrix of the score vector, is referred to as the \textit{observed information matrix}. We
will denote the matrix as $\I_O$.
\begin{equation}
\I_O=  -
\begin{pmatrix}
 \frac{\partial \ell_R}{\partial \sigma^2\partial \sigma^2} &
 \frac{\partial \ell_R}{\partial \sigma^2 \partial \kappa_1} &
 \cdots &\frac{\partial \ell_R}{\partial \sigma^2 \partial \kappa_m}
 &
 \\
\frac{\partial \ell_R}{\partial \kappa_1\partial \sigma^2 }
&\frac{\partial \ell_R}{\partial \kappa_1\partial \kappa_1} & \cdots
&
\frac{\partial \ell_R}{\partial \kappa_1\partial \kappa_m} \\
\vdots

 & \vdots & \ddots & \vdots
 \\
\frac{\partial \ell_R}{\partial \kappa_m\partial \sigma^2}

 &  \frac{\partial \ell_R}{\partial \kappa_m\partial \kappa_1} &
 \cdots& \frac{\partial \ell_R}{\partial \kappa_m\partial \kappa_m}
\end{pmatrix}.
\label{eq:IO}
\end{equation}

An standard approach to find zeros of the score vector is the Newton-Raphson method in Algorithm~\ref{alg:NR} which requires the observed information.
Unfortunately, we shall see that computing the observed information is computationally prohibitively for large data sets.
Therefore, the \emph{expected information}---the \emph{Fisher information}---is often preferred due to its simplicity. The resulting algorithm
is referred to as the \textit{Fisher-scoring algorithm} \cite{jenn76,L87}.
The Fisher scoring algorithm is a success in simplifying the approximation of the Hessian
matrix of the log-likelihood. Still, evaluating
of the elements of the Fisher information matrix is one of the bottlenecks in a maximizing
log-likelihood procedure, which prohibits the Fisher scoring algorithm for larger data sets.
In particular, the high-throughput technologies in the biological science, recommendation systems,
engineering and social network mean that the size of data sets and the corresponding statistical models have suddenly
increased by several orders of magnitude. Further reducing computations is deserved in
large scale statical models like genome-wide association studies.

In \cite{ZGL16}, 
the authors prove that the mean of the observed information, $\I_O$, and Fisher information, $\I_F$, can be split into two parts:
\begin{equation}
\frac{\I_O+\I_F}{2} = \I_A+\I_Z.
\end{equation}
The expectation of the first part, $\I_A$ is equivalent to the Fisher information which maintains the essential information on the variance parameters while possesses a much simper form than the Fisher information; the second part $\I_Z$ involves a lot of computations is an random zero matrix.
The approximated information $\I_A$ is much simper than information used in \cite{lis12,zhang10}, \cite[eq.6, eq.7]{zhou12}. And such an approximation significantly reduce computations and speed up the linear mixed model \cite{WZW13}.


The aim  of this paper is to provide detailed derivation of such an information splitting formula and supplies an self-contained background information. The remaining of the paper is organised as follows. In \S~\ref{sec:prel} we shall introduce the restricted maximum likelihood method for linear mixed model and derives the formula for the restricted log-likelihood and its scores. In \S~\ref{sec:info}, we shall derive the observed, Fisher, and the averaged splitting information. In \S~\ref{sec:comp} we shall derive computational friendly formulas for evaluating elements of the averaged splitting information matrix. We concludes the paper with some discussion in the last section.

\begin{algorithm*}[!t]
\caption{Newton-Raphson/Fisher Scoring/Averaged Information Splitting method to solve {$S(\theta)=0$.}}
\begin{algorithmic}[1]
\State {Give an initial guess of $\theta_0$} \For{ $k=0, 1, 2,
\cdots$ until convergence }
 \State{Solve $$ \begin{cases}
 \I_O(\theta_k) \delta_k=S(\theta_k)     & \text{ for Newton-Raphson} \\
\I_F(\theta_k)  \delta_k=S(\theta_k)     & \text{ for Fisher scoring} \\
\I_A(\theta_k)  \delta_k=S(\theta_k)     & \text{ for average information splitting}
\end{cases}$$}
 \State{$\theta_{k+1}=\theta_k+\delta_k$}
\EndFor
\end{algorithmic}
\label{alg:NR}
\end{algorithm*}

\section{restricted log-likelihoods and its score}
\label{sec:prel}
The restricted maximum likelihood method was introduced by Patterson and Thompson \cite{PT71}. The aim of the method was to reduce bias in animal breeding models with unbalanced block structures. The maximum likelihood estimates of variance parameters in a linear mixed model have large bias. We shall fist use the simple linear model to articulate the difference between an ML estimation and a REML estimation of the variance parameters.

Consider the simplest linear model
\begin{equation}
y=X \tau +\epsilon, \quad \epsilon \sim N(0, \sigma^2 I). \label{eq:LM}
\end{equation}
The bias between an maximum likelihood estimation for the variance parameter, $\hat{\sigma}^2$ and $\sigma^2$ is
\begin{equation}
\mathrm{Bias}(\hat{\sigma}^2, \sigma^2)=E(\hat{\sigma}^2) -\sigma^2 = \frac{p}{n}\sigma^2, \label{eq:MLBIAS}
\end{equation}
where $p=\mathrm{rank}(X)$ is the number of fixed effects and $n$ is the number of observations. When the observations is small and the number of fixed effects is relative large to the observations. the bias $\mathrm{Bias}(\hat{\sigma}^2, \sigma^2)$ is relative large. Such situations happen when one subdivide a big data set into many small groups, and view each group as an individual block.
\subsection{Bias of ML estimation on variance parameter}
The individual observations
$y_i$, $i=1,\ldots,n$, in the linear model \eqref{eq:ML} are statistically independent and have
distribution $y_i \sim N(\bfx_i \tau, \sigma^2)$, where $\bfx_i$
is the \textit{i}th row of $X$. The likelihood for the joint
distribution is defined as
\begin{align*}
&L(\tau, \sigma^2; y)= \prod _{i=1}^{n}f(y_i; \tau, \sigma^2)
=\prod_{i=1}^n
\frac{\exp\left(-\frac{(y_i-x_i\tau)^2}{2\sigma^2}\right)}{\sqrt{2\pi\sigma^2}}\\
&=\left(\frac{1}{\sqrt{2\pi\sigma^2}}\right)^n
\exp\left(-\frac{(y-X\tau)^T(y-X\tau)}{2\sigma^2}\right).
\end{align*}

The log-likelihood function, $\ell$, is given by
\begin{align*}
\ell&=\log L(\tau, \sigma^2; y) \\
&=-\frac{n}{2}\log(2\pi \sigma^2)-\frac{(y-X\tau)^T(y-X\tau)}{2\sigma^2}
\\
&=-\frac{n\log(2\pi)}{2}-\frac{n\log(\sigma^2)}{2}\\
& -\frac{1}{\sigma^2}(y^Ty-2y^TX\tau+\tau^TX^TX\tau).
\end{align*}
The \textit{score  functions} of $\tau$ and $\sigma^2$ are the first derivatives of $\ell$ with respect to $\tau$ and $\sigma^2$ respectively
\begin{align*}
s(\tau):&=\frac{\partial \ell}{\partial \tau}
= \frac{1}{\sigma^2}(X^Ty -X^TX\tau),
\\
s(\sigma^2):&=\frac{\partial \ell}{\partial \sigma^2}
=\frac{-n}{2\sigma^2}+\frac{1}{2\sigma^4}(y-X\tau)^T(y-X\tau).
\end{align*}
The maximum likelihood estimation for the fixed effects,
$\hat{\tau}$, and the variance parameter $\hat{\sigma}^2$
satisfy that
\begin{equation}
\begin{cases}
s(\hat{\tau})=\frac{\partial \ell}{\partial \tau}\big\vert_{\hat{\tau}}=0,  \\
s(\hat{\sigma}^2)=
\frac{\partial \ell}{\partial \sigma^2}\big\vert_{\hat{\sigma}^2}=0.
\end{cases}
\label{eq:eML}
\end{equation}
For this simple model, the ML estimation is obtained easily:
\begin{equation*}
\begin{cases}
\hat{\tau}=(X^TX)^{-1}X^Ty,\\
\hat{\sigma^2}=\frac{1}{n}(y^T(I-P_X)y):=\frac{S_R}{n},
\end{cases}
\end{equation*}
where
$
P_X=X(X^TX)^{-1}X^T  \label{eq:px}
$
is the \textit{projection matrix} for $X$ and
\begin{align}
S_R:&=(y-X\hat{\tau})^T(y-X\hat{\tau})=y^T(I-P_X)y.
\label{eq:SR}
\end{align}
is the \textit{residual sum of squares}. Since
\begin{equation}
E(y^T(I-P_X)y) 
 =\tr((I-P_X)E(yy^T))
=(n-p)\sigma^2,  \label{eq:ESR}
\end{equation}
where $p=\operatorname{rank}(X)$, we have
\begin{equation}
E(\hat{\sigma}^2)=\frac{E(y^T(I-P)y)}{n}=\frac{n-p}{n}\sigma^2.
\label{eq:ML}
\end{equation}
Without difficulty, one can obtain the bias between the ML estimation $\hat{\sigma}^2$ and $\sigma^2$ as in \eqref{eq:MLBIAS}.

\subsection{REML estimation for linear models}

In the framework of REML, the observation $y$ is divided into two (orthogonal) components: one of the
component of $y$ contains all the (fitted) residual error information
in the linear mixed model \eqref{eq:LM}. Employing the maximum likelihood on the two orthogonal components results in two smaller problems (compared with the ML estimation). The partition is constructed based on the following lemma
\begin{lemma}
Let $X\in \mathbb{R}^{n\times p}$ be full rank and the projection matrix $P_X=X(X^TX)^{-1}X$, then there exist an orthogonal matrix $K=[K_1, K_2]$ such that
\begin{enumerate}
\item $P_X=K_1K_1^T$;
\item $K_2^TX=0$;
\item $I-P_X=K_2K_2$.
\end{enumerate}
\label{lemma:1}
\end{lemma}
\begin{proof}
Since $P_X$ is a projection matrix with rank $p$, the matrix is symmetric with eigenvalues 1 and 0. Then there exist an eigenvalue decomposition as
\[
P_X=(K_1, K_2) \begin{pmatrix}
I_p & 0 \\
0 & 0
\end{pmatrix}\begin{pmatrix}
K_1^T \\
K_2^T
\end{pmatrix}
 =K_1K_1^T.
\]
Therefore, $K_1 \in V_1$, and $K_2 \in V_0$ where $V_i$, $i=0,1$ is the eigenspace which corresponds to the eigenvalue $i$.
For a projection matrix, we have $P_XX=X$, therefore, $X\in V_1$ and $K_2^TX=0$ because eigenvectors which associate to different eigenvalues are orthogonal. The term 3) follows because
\begin{align*}
I-P_X= [K_1 K_2] \begin{pmatrix}K_1^T \\K_2^T
\end{pmatrix} -K_1K_1^T =K_2K_2^T.
\end{align*}
\end{proof}
Let $K=[K_1, K_2]$ be an orthogonal matrix such that
$P_X=K_1K_1^T$ and $K_2^TX=0$. Let
$y_i=K_i^Ty$, $i=1,2$, then $E(K^Ty)=K^TE(y)$, and
$$\mathrm{var}(K^Ty)=K^T\mathrm{var}(y)K=\sigma^2 I.$$
\begin{equation}
\begin{pmatrix}
y_1\\ y_2
\end{pmatrix}
\sim N\left(\begin{pmatrix} K_1^TX\tau \\ 0 \end{pmatrix},
\sigma^2\begin{pmatrix} I_p & 0\\ 0 &I_{n-p}
\end{pmatrix}\right).
\end{equation}
Apply the maximum likelihood methods to $y_1$ and $y_2$, we obtain
two likelihood functions
\begin{align}
\ell_1(\tau, \sigma^2;y_1) & =-\frac{p}{2}\log(2\pi)-\frac{p}{2}\log(\sigma^2)-
 \nonumber \\
 & \frac{(y_1-K_1^TX\tau)^T(y_1-K_1^TX\tau)}{2\sigma^2}, \\
\ell_R =\ell_2(\sigma^2; y_2) &= -\frac{n-p}{2}\log(2\pi\sigma^2)
+\frac{y_2^Ty_2}{\sigma^2}.
\end{align}

The estimation of $\tau$ based on the likelihood function $\ell_1$ is the same as
the maximum likelihood estimation:
\begin{align*}
\hat{\hat{\tau}}=(X^T\underbrace{K_1K_1^T}_{=P_X}X)^{-1}X^T\underbrace{K_1K_1^T}_{=P_X}y
=(X^TX)^{-1}X^Ty=\hat{\tau}.
\end{align*}

Restrict the maximum likelihood method to the marginal distribution of the residual
component, $y_2$, one can obtain the estimation of the variance parameters. The marginal log-likelihood function, $\ell_2$, does not depend on $\tau$.
The estimation of $\sigma^2$ in the marginal likelihood is
\begin{equation}
\hat{\hat{\sigma}}^2= \frac{y_2^Ty_2}{n-p}=\frac{y^TK_2K_2^Ty}{n-p}
 =\frac{y^T(I-P_X)y}{n-p}=\frac{S_R}{n-p}, \label{eq:REMLsigma}
\end{equation}
where $S_R$ is the residual sum of squares defined in
\eqref{eq:SR}.
Employ the result in \eqref{eq:ESR}, we conclude that the REML
estimation for the variance parameter is unbiased,
$$E(\hat{\hat{\sigma}}^2)=\sigma^2.$$
Therefore, the REML estimation on variance parameter in the linear model \eqref{eq:ML} is unbiased. In summary, REML has twofold functionalities. On one hand, one can use the REML as an approach to model reduction, reducing the problem size. On the other hand, it reduces the bias for the estimation of the variance parameter.

\subsection{Restricted log-likelihood for linear mixed model}

 The construction of the REML
estimation for \eqref{eq:LM} with general (co-)variance
structure is constructed based on the following fact.
\begin{lemma}
Let $X\in \mathbb{R}^{n\times p}$ be full rank with $p<n$, then there exists an nonsingular matrix $L=[L_1, L_2]$ such that
\begin{enumerate}
\item $L_1^TX=I_{p\times p}$;
\item $L_2^TX=0$;
\item $I-P_X=I-X(X^TX)^{-1}X^T=L_2(L_2^TL_2)^{-1}L_2^T$.
\end{enumerate}
\label{lemma:2}
\end{lemma}
\begin{proof}
Let $B \in \mathbb{R}^{(n-p)\times (n-p)}$ be any nonsingular matrix and $K_2K_2^T=I-P_X$ as in Lemma~\ref{lemma:1}. Then $BK_2^TX=0$ and $\rank(K_2B^T)=n-p$. Therefore, $\{X, K_2B^T \}$ forms a set of basis of $\mathbb{R}^{n\times n}$. Denote $L^T=[X K_2B^T]^{-1}$, then $L^T[X, K_2B^T]=I$, we have
\begin{equation}
\begin{pmatrix}
L_1^TX  & L_1^TK_2B^T \\
L_2^TX  & L_2^TK_2B^T
\end{pmatrix}
=
\begin{pmatrix}
I_{p\times p} & 0 \\
0 & I_{(n-p) \times (n-p)}
\end{pmatrix}.
\end{equation}
This gives that $L_1^T X=I_{p\times p}$ and $L_2^TX=0_{(n-p)\times p}$.

To prove 3), let $D=[X,L_2]$, then the columns of $D$ forms a basis set of $\mathbb{R}^{n\times n}$ and $D$ is an nonsingular matrix.
Apply the identity $I=DD^{-1}D^{-T}D=D(D^TD)^{-1}D^T$ and write it in a block matrix multiplication forms
\[
 I=(X, L_2) \begin{pmatrix}
X^TX & X^TL_2 \\ L_2^TX & L_2^TL_2
\end{pmatrix}^{-1}
\begin{pmatrix}
X^T  \\ L_2^T
\end{pmatrix}.
\]
Employ the fact the $L_2^TX=0$, and multiply the block matrices on the right hand side, we have
$$
I=P_X+L_2(L_2^TL_2)^{-1}L_2^T.
$$
\end{proof}
\begin{corollary}
Let $X\in \mathbb{R}^{n\times p}$ be full rank with $p<n$, $H$ be a positive definite matrix, and $L=[L_1,L_2]$ be the linear transform matrix in Lemma~\ref{lemma:2} such that
$L_1^TX=I_{p\times p}$ and $L_2^TX=0_{(n-p) \times p}$.  If $P$ is defined in \eqref{eq:p}
$$
P=H^{-1} -H^{-1}X(X^TH^{-1}X)^{-1}X^TH^{-1}
$$
then we have
\begin{equation}
P=L_2(L_2^THL_2)^{-1}L_2^T. \label{eq:L2HL2}
\end{equation} \begin{equation}
(XH^{-1}X)^{-1}=L_1^THL_1-L_1^THL_2(L_2^THL_2)^{-1}L_2^THL_1. \label{eq:XHX}
 \end{equation}

\end{corollary}
\begin{proof}
Since $H$ is symmetric positive definite, then there exists a symmetric positive definite matrix $H^{1/2}$, Let $\hat{X}=H^{-1/2}X$, then $\hat{X}$ is full rank and $L_2^TH^{1/2}\hat{X}=0$. According to Lemma~\ref{lemma:2}, we have
\[
I-\hat{X}(\hat{X}^T\hat{X})^{-1}\hat{X}=H^{1/2}L_2(L_2^THL_2)^{-1}L_2^T
\]
Multiply $H^{-1/2}$ on both sides, we obtain \eqref{eq:L2HL2}. For \eqref{eq:XHX}, we have
\begin{align*}
 & L_1^THL_1 -L_1^TH \underbrace{L_2(L_2^THL_2)^{-1}L_2^T}_{P}HL_1 \\
 & = L_1^THL_1 -L_1^T(H-X(X^TH^{-1}X)^{-1}X^T)L_1  \\
 & =\underbrace{L_1^TX}_{I_{p\times p}}(X^TH^{-1}X)^{-1}\underbrace{X^TL_1}_{I_{p\times p}}=(X^TH^{-1}X)^{-1}.
\end{align*}
\end{proof}

For $X\in
\mathbb{R}^{n\times p}$, Let $L=[L_1, L_2]$ be the linear transformation defined in Lemma~\ref{lemma:2}.
such that $L_1^TX=I_p$ and $L_2^TX=0$. Use this transform, we obtain
\begin{equation}
L^Ty=\begin{pmatrix}
y_1 \\ y_2
\end{pmatrix}
\sim
N\left( \begin{pmatrix}
\tau \\0
\end{pmatrix}, \sigma^2
\begin{pmatrix}
L_1^THL_1 & L_1^THL_2 \\
L_2^THL_1 & L_2^THL_2
\end{pmatrix}
 \right).
 \label{eq:y1y2}
\end{equation}
According to the result in \cite[p40, Thm
2.44]{all04}, the marginal distribution of $y_2$ is given as
$y_2\sim N(0,\sigma^2L_2^THL_2).$
The associated likelihood function corresponding to $y_2$ is
\begin{multline}
\ell_R= \ell_2=\ell(\sigma^2,\phi;y_2)
=-\frac{1}{2}\{(n-p)\log(2\pi \sigma^2)   \\
+\log|L_2^THL_2|+y_2^TL_2(L_2^THL_2)^{-1}L_2^Ty_2)/\sigma^2 \}.
\label{eq:l2}
\end{multline}
This form is equivalent to \eqref{eq:ellR}.
\begin{theorem}
The residual log-likelihood for the linear model in
\eqref{eq:l2} is equivalent to
\begin{align}
\ell_{R} =& -\frac{1}{2}\left\{ (n-p)\log ( \sigma^2)+\log |H|
+\log|X^TH^{-1}X| \right\}   \nonumber \\
 &-\frac{1}{2}y^TPy/\sigma^2 +\mathrm{const} .
\end{align}
where
\begin{equation}
P=H^{-1}-H^{-1}X(X^TH^{-1}X)^{-1}X^TH^{-1}. \label{eq:PHX}
\end{equation}

\end{theorem}
\begin{proof}
 First we notice that $P=L_2(L_2^THL_2)^{-1}L_2$ \eqref{eq:L2HL2}), and \eqref{eq:XHX}.
Then we use the identity
\begin{align*}
 & \left\vert \begin{pmatrix}
I_p & -(L_1^THL_2)(L_2^THL_2)^{-1} \\
0 &I_{n-p}
\end{pmatrix}
\begin{pmatrix}
L_1^THL_1 & L_1^THL_2 \\
L_2^THL_1 & L_2^THL_2
\end{pmatrix} \right\vert \\
& =\left\vert
\begin{pmatrix}
(X^TH^{-1}X)^{-1} & 0\\
L_2^THL_1 & L_2^THL_2
\end{pmatrix}\right\vert=\vert L^THL \vert,
\end{align*}
We have $|L^THL|=|H||L^TL|=|(X^TH^{-1}X)^{-1}||L_2^THL_2|$ and
\begin{equation}
\log|L^TL|+\log|H|=\log|L_2^THL_2^T|-\log|X^TH^{-1}X|.
\end{equation}
Note the construction of $L$ does not depend on $\sigma^2$ and
$\phi$, therefore $\log|L^TL|$ is a constant. 
\end{proof}

\subsection{The score functions for the restricted log-likelihood}

\label{sec:comput}
To derive the scores for the restricted log-likelihood, we shall use the formula \eqref{eq:l2} rather than \eqref{eq:ellR}. In general setting, when we present the restricted log-likelihood, we shall use \eqref{eq:ellR}, because it does not involve any intermediate variables like $L_2$.

\begin{lemma}
Let $A(\kappa)$ be a nonsingular parametric matrix, then we have

 \begin{align} \frac{\partial \log \lvert A \rvert }{\partial \kappa_i} &= \tr{A^{-1} \frac{A}{\partial \kappa_i}}, \label{eq:logdetA} \\
\frac{\partial A^{-1}}{ \partial \kappa_i} & = -A^{-1} \frac{\partial A}{ \partial \kappa_i} A^{-1} . \label{eq:A-1}
\end{align}

\end{lemma}
\begin{IEEEproof}
See \cite[p.305, eq.8.6]{Har97} for \eqref{eq:logdetA} and \cite[p.307, eq. 8.15]{Har97} for \eqref{eq:A-1}
\end{IEEEproof}

\begin{theorem}[\cite{GTC95}]
Let $X\in \mathbb{R}^{n\times p}$ be full rank. The scores of the residual
log-likelihood  the linear model in \eqref{eq:LM} is given by
\begin{align}
s(\sigma^2)&=\frac{\partial \ell_R}{\partial \sigma^2}
=-\frac{1}{2}\left\{
\frac{n-p}{\sigma^2}-\frac{y^TPy}{\sigma^4}\right \}, \\
s(\kappa_i)&=\frac{\partial \ell_R}{\partial
\kappa_i}=-\frac{1}{2}\left\{ \tr(P \dot{H}_i) -
\frac{1}{\sigma^2}y^TP \dot{H}_iP y  \right\},
\end{align}
where $P$ is defined in \eqref{eq:PHX}, and $\dot{H}_i=\frac{\partial{H}_i}{\partial \kappa_i}$
\end{theorem}
\begin{proof}
Consider the residual loglikelihood function in \eqref{eq:l2},
it follows that $s(\sigma^2)=\frac{\partial \ell_R}{\partial
\sigma^2}$.
\begin{equation}
s(\kappa_i)= -\frac{1}{2} \left\{  \frac{\partial\log |L_2^THL_2|}{\partial \kappa_i}+
\frac{1}{\sigma^2}\frac{\partial (y^TPy)}{\partial \kappa_i} \right \}.
\label{eq:sk}
\end{equation}
 Using the fact on matrix derivatives of log determinant in \eqref{eq:logdetA}
and the property of the trace
operation $\tr(AB)=\tr(BA)$
\begin{align}
&\frac{\partial \log(\lvert L_2^THL_2\rvert)}{\partial \kappa_i} =
\tr\left((L_2HL_2)^{-1} \frac{\partial (L_2^THL_2)}{\partial
\kappa_i}\right)
\nonumber \\
&=\tr\left( \underbrace{L_2(L_2^THL_2)^{-1} L_2}_{=P}
\dot{H}_i\right)=\tr\left(P\dot{H}_i\right).
\end{align}
 One easy way to calculate the second term in \eqref{eq:sk} is
 to use the relationship
$$ P=L_2(L_2^THL_2)^{-1}L_2^T$$
and the result on matrix derivatives of an inverse matrix
\eqref{eq:A-1}
\begin{equation*}
\frac{\partial H^{-1}}{ \partial \kappa_i}
=-H^{-1}\frac{\partial H}{\partial \kappa_i} H^{-1}.
\end{equation*}
We have
\begin{align}
 &\frac{\partial (L_2(L_2^THL_2)^{-1}L_2^T)}{\partial
\kappa_i}=L_2\frac{\partial (L_2^THL_2)^{-1}}{\partial
\kappa_i}L_2^T  \nonumber \\
& =-L_2(L_2^THL_2)^{-1}\frac{\partial (L_2^THL_2)}{\partial
\kappa_i} (L_2^THL_2)^{-1}L_2^T \nonumber \\
& =-\underbrace{L_2(L_2^THL_2)^{-1}L_2^T}_{=P}\dot{H}_i
\underbrace{L_2(L_2^THL_2)^{-1}L_2^T}_{=P}
\nonumber \\& =-P\dot{H}_iP
=\dot{P}_i \label{eq:dpxh}.
\end{align}
\end{proof}

\section{Derive the information matrices}

\label{sec:info}

As the score vector for the restricted log-likelihood available, the observed information and the Fisher information can be derived by definition with the help of some matrix algebra operations.
\begin{lemma}
Let $y\sim N(X\tau, \sigma^2H)$
be a random variable
and $H$ is
symmetric positive definite matrix,where $\operatorname{rank}(X)=\nu$, then
$$P=H^{-1}-H^{-1}X(X^TH^{-1}X)^{-1}XH^{-1}$$ is a weighted projection
matrix such that
\begin{enumerate}
  \item $PX=0$;
  \item $PHP=P$;
  \item $\tr(PH)=n-\nu$;
  \item $PE(yy^T)=\sigma^2PH$.
\end{enumerate}
\label{lem:P}
\end{lemma}
\begin{IEEEproof}
The first 2 terms can be verified by directly by computation. Since
H is a positive definite matrix, there exist $H^{1/2}$ such that
\begin{align*}
\tr(PH) & =\tr(H^{1/2}PH^{1/2})=\tr(I-\hat{X}(\hat{X}^T\hat{X})^{-1}
\hat{X}) \\
& =n-\rank(\hat{X})=n-\nu.
\end{align*}
where $\hat{X}=H^{-1/2}X$. The 4th item follows because
\begin{align*}
P
E(yy^T)=P(\mathrm{var}(y)-X\tau
(X\tau)^T)  \\
=\sigma^2PH-PX\tau(X\tau)^T=\sigma^2PH.
\end{align*}
\end{IEEEproof}

\begin{lemma}
Let $H$ be a parametric matrix of $\kappa$, and $X$ be an constant matrix, then the partial derivative of
the projection matrix $$ P=H^{-1}-H^{-1}X(X^TH^{-1}X)^{-1}XH^{-1} $$
with respect to $\kappa_i$ is
\begin{equation}
\dot{P}_i=-P\dot{H}_iP,
\end{equation}
where $\dot{P}_i=\frac{\partial P}{\partial \kappa_i}$ and $\dot{H}_i=\frac{\partial H}{\partial \kappa_i}.$
\label{lem:PD}
\end{lemma}
\begin{IEEEproof}
Using the derivatives of the inverse of a matrix \eqref{eq:A-1},
we have
\begin{align*}
\dot{P}_i =& \frac{\partial}{\partial \kappa_i}(H^{-1}-H^{-1}X(X^TH^{-1}X)^{-1}X^TH^{-1}) \\
  = & -H^{-1}\dot{H}_iH^{-1}+H^{-1}\dot{H}_iH^{-1}X(X^TH^{-1}X)^{-1}X^TH^{-1} \\
& -H^{-1}X(X^TH^{-1}X)^{-1}X^TH^{-1}\dot{H}_i \times \\
& \quad H^{-1}X(X^TH^{-1}X)^{-1}X^TH^{-1} \\
& + H^{-1}X(X^TH^{-1}X)^{-1}X^TH^{-1}\dot{H}_iH^{-1} \\
=&-H^{-1}\dot{H}_i+H^{-1}X(X^TH^{-1}X)^{-1}X^TH^{-1}\dot{H}_iP\\
=&-P\dot{H}_iP.\\
\end{align*}
\end{IEEEproof}

\subsection{Formulas for the observed information}
\begin{theorem}
The element of the observed information matrix for the residual
log-likelihood \eqref{eq:ellR} is given by
\begin{align}
\I_O(\sigma^2,\sigma^2) &=
\frac{y^TPy}{\sigma^6}-\frac{n-p}{2\sigma^4}, \label{eq:ISS} \\
\I_O(\sigma^2,\kappa_i) &=
\frac{1}{2\sigma^4}y^TP\dot{H}_iPy , \label{eq:ISK}\\
\I_O(\kappa_i,\kappa_j) & = \frac{1}{2}\left\{\tr(P\dot{H}_{ij})-
\tr(P\dot{H}_iP \dot{H}_j)
 \right\}  \nonumber  \\
 & +\frac{1}{2\sigma^2}\left\{ 2y^TP\dot{H}_iP\dot{H}_j Py -y^T
 P\ddot{H}_{ij}Py\right\}.
\label{eq:IKK}
\end{align}
where $\dot{H}_i=\frac{\partial H}{\partial \kappa_i}$,
$\ddot{H}_{ij}=\frac{\partial^2 H}{\partial K_i \partial K_j}$.
\end{theorem}
 \begin{proof}
 The result in \eqref{eq:ISS} is standard according to the definition. The result in
 \eqref{eq:ISK} follows according to the result in Lemma \ref{lem:PD}. If one uses the score in \eqref{eq:ssigma}.
 The first term in \eqref{eq:IKK} follows because
 \begin{align*}
 \frac{\partial \tr(P\dot{H}_i)}{\partial
 \kappa_j}&=tr(P\ddot{H}_{ij})+\tr(\dot{P}_j\dot{H}_i) \quad ( \dot{P}_j=-PH_jP)\\
 &=\tr(P\ddot{H}_{ij})-\tr(P\dot{H}_jP\dot{H}_i).
 \end{align*}
 The second term in \eqref{eq:IKK} follows according to the result in Lemma \ref{lem:PD}.
 \begin{equation}
 -\frac{\partial (P\dot{H}_iP)}{\partial
 \kappa_j}=P\dot{H}_jP\dot{H}_iP-P\ddot{H}_{ij}P+P\dot{H}_iP\dot{H}_jP.
 \end{equation}
 Further note that $\dot{H}_i$, $\dot{H}_j$ and $P$ are symmetric.
The second term in \eqref{eq:IKK} follows because of
 $$y^TP\dot{H}_iP\dot{H}_jPy=y^TP\dot{H}_jP\dot{H}_iPy. $$

 \end{proof}
\subsection{Formulas of the Fisher information matrix}

\begin{table*}[t!]
\centering
 \caption{comparison between the observed, Fisher and averaged splitting information}
 \label{tab:splitting}
\begin{tabular}{llll}
  \hline
  index & $ \I_O$ & $\I $ & $\I_A$  \\
  $(\sigma^2, \sigma^2)$ &
  $\frac{y^TPy}{\sigma^6}-\frac{n-\nu}{2\sigma^2}$ & $\frac{n-\nu}{2\sigma^4}$
  & $\frac{y^TPy}{2\sigma^6}$  \\
  $(\sigma^2,\kappa_i)$ & $\frac{y^TPH_iPy}{2\sigma^4}$&
  $\frac{\tr(PH_i)}{2\sigma^2}$ &$\frac{y^TPH_iPy}{2\sigma^4}$\\
  $(\kappa_i, \kappa_j)$ & $\I_O(\kappa_i,\kappa_j) $
  & $\frac{\tr(PH_iPH_j)}{2}$
  & $\frac{y^TPH_iPH_jPy}{2\sigma^2}$\\
 \\
& \multicolumn{3}{c}{ $\I_O(\kappa_i,\kappa_j) = \frac{\tr(PH_{ij})- \tr(PH_iP H_j)}{2} +\frac{2y^TPH_iPH_j Py -y^T
    PH_{ij}Py}{2\sigma^2}, $ }\\
& $H_i=\frac{\partial H_i}{\partial \kappa_i}$, &
$H_{ij}=\frac{\partial^2H}{\partial H_i\partial H_j}$ \\
  \hline
\end{tabular}
\end{table*}

The \textit{Fisher information matrix}, $\I$, is the
expect value of the observed information matrix, $ \I=E(\I_O).$
According to such a definition, with some calculation, we have
\begin{theorem}
The element of the Fisher information matrix for the residual
log-likelihood function in \eqref{eq:ellR} is given by
\begin{align}
\I(\sigma^2,\sigma^2) &=E(\I_O(\sigma^2,\sigma^2))=\frac{\tr(PH)}{2\sigma^4}=\frac{n-\nu}{2\sigma^4}, \label{eq:FISS}\\
\I(\sigma^2, \kappa_i)
&=E(\I_O(\sigma^2,\kappa_i))=\frac{1}{2\sigma^2}
\tr(P\dot{H}_i),\label{eq:FISK} \\
\I(\kappa_i, \kappa_j)
&=E(\I_O(\kappa_i,\kappa_j))=\frac{1}{2}\tr(P\dot{H}_iP\dot{H}_j)
.\label{eq:FIKK}
\end{align}
\end{theorem}
\begin{proof} The formulas can be found in \cite{GTC95}.
 Here we supply alternative proof.
First note that $PX=0$ and
\begin{align}
PE(yy^T)&=P(\sigma^2H + X\tau (X\tau)^T)=\sigma^2PH.
\label{eq:PEyy}
\end{align}
Then
\begin{align}
E(y^TPy)&=E(\tr(Pyy^T))=\tr(PE(yy^T))  \nonumber \\
& =\sigma^2\tr(PH) =(n-\nu)\sigma^2. \label{eq:Eypy}
\end{align}
where $\rank(L_2)=n-\rank(X)$ due to $L_2^TX=0$.
Therefore
\begin{equation}
E(\I_O(\sigma^2,\sigma^2))=\frac{E(y^TPy)}{\sigma^6}-\frac{n-\nu}{2\sigma^4}=\frac{n-\nu}{2\sigma^4}.
\end{equation}
Second, we notice that $PHP=P$. Apply the procedure in
\eqref{eq:Eypy}, we have
\begin{align}
E(y^TP\dot{H}_iPy)&= \tr(P\dot{H}_iPE(yy^T))=\sigma^2\tr(P\dot{H}_iPH)\nonumber \\
&=\sigma^2\tr(PHP\dot{H}_i) =\sigma^2\tr(P\dot{H}_i),  \label{eq:yphpy}\\
E(y^TP\dot{H}_iP\dot{H}_jPy) &=
\sigma^2\tr(P\dot{H}_iP\dot{H}_jPH)  \nonumber \\
&=\sigma^2\tr(PHP\dot{H}_iP\dot{H}_j) \nonumber \\
& =\sigma^2\tr(P\dot{H}_iP\dot{H}_j), \label{eq:yphphpy}\\
E(y^TP\ddot{H}_{ij}Py)&=\sigma^2\tr(P\ddot{H}_{ij}PH)=\sigma^2\tr(P\ddot{H}_{ij}).
\label{eq:yphhpy}
\end{align}
Substitute \eqref{eq:yphpy} into \eqref{eq:ISK}, we obtain
\eqref{eq:FISK}.   Substitute \eqref{eq:yphphpy} and
\eqref{eq:yphhpy} to \eqref{eq:IKK}, we obtain \eqref{eq:FIKK}.
\end{proof}

Because the elements of the Fisher information
matrix have simper forms than these of the observed information
matrix, in practice, the
\textit{Fisher information matrix}, $\I=E(\I_O)$, is preferred.
The corresponding algorithm is referred to as
the \textit{Fisher scoring algorithm} \cite{L87}. The Fisher scoring algorithm is widely used in many machine learning algorithms.

The Fisher scoring algorithm is a great success in reducing computations in the Hessian matrix of a log-likelihood. However, notice that the elements $\I(\sigma^2, \kappa_i)$ and $\I(\kappa_i, \kappa_j)$ in the Fisher information still involve computationally intensive trace terms of matrix products.
Evaluating these trace terms is still computationally prohibitive for big data sets.  On the other hand, we notice that
some quadratic terms in $\I_O(\sigma^2, \kappa_i)$ and $\I_O(\kappa_i, \kappa_j)$ is easier to be evaluated because they can be transformed as several matrix vector multiplications.
One natural thinking is whether one can obtain an approximated information matrix by some combination of the Fisher information and the observed information such that only quadratic terms remain.

Following the idea used in \cite{GTC95}, where the \textit{average information} $\I_{A}$ is introduced as
\begin{align}
\I_A(\sigma^2,\sigma^2)&=\frac{1}{2\sigma^6}y^TPy;  \label{eq:ASS}\\
\I_A(\sigma^2,\kappa_i)&=\frac{1}{2\sigma^4}y^TP\dot{H}_iPy;\label{eq:ASK}\\
\I_A(\kappa_i,\kappa_j)&=\frac{1}{2\sigma^2}y^TP\dot{H}_iP\dot{H}_jPy.
\label{eq:AKK}
\end{align}
The authors in \cite{ZGL16} prove that such an ``average information" is in fact a part of the mean of the observed and Fisher information.
For example,
\begin{multline}
\frac{\I(\kappa_i, \kappa_j)+\I_O(\kappa_i,\kappa_j)}{2}
=\underbrace{\frac{y^TPH_iPH_jPy}{2\sigma_2}}_{\I_A(\kappa_i,\kappa_j)} \\
+\underbrace{\frac{\tr(PH_{ij})-y^TPH_{ij}Py/\sigma^2
}{4}}_{\hat{\I}_Z(\kappa_i,\kappa_j)}.
\end{multline}
According to the classical matrix splitting techniques
\cite[p.94]{Varga99}, this technique can be formulated as follows.

\begin{theorem}
Let $\I_O$ and $\I$ be the observed information matrix and the Fisher information matrix for the residual
log-likelihood of linear mixed model respectively, then the average of the observed information matrix and the Fisher information matrix can be split as $\frac{\I_O+\I}{2}=\I_A+ I_Z$, such that the expectation of $\I_A$ is the Fisher information matrix and $E(\hat{\I}_{Z})=0$.  \label{cor1}
\end{theorem}
\begin{proof}
Let the element of $\I_A$ be defined as in \eqref{eq:ASS} to \eqref{eq:AKK}, we have
\begin{align}
I_Z(\sigma^2,\sigma^2)&=0, \\
I_Z(\sigma^2,\kappa_i)&=\frac{tr(P\dot{H}_i)}{4\sigma^2}-\frac{y^TP\dot{H}_iPy}{4\sigma^4}, \\
I_Z(\kappa_i,\kappa_j)& =\frac{\tr(PH_{ij})-y^TPH_{ij}Py/\sigma^2,
}{4}
\end{align}
Apply the result in \eqref{eq:Eypy}, we have
\begin{equation}
E(\I_A(\sigma^2,\sigma_2))=\frac{(n-\nu)}{2\sigma^4}=\I(\sigma^2,\sigma^2).
\end{equation}
Apply the result in \eqref{eq:yphpy}, we have $$ E(I_Z(\sigma^2,\kappa_i))=0$$ and
\begin{equation*}
E(\I_A(\sigma^2,\kappa_i))=\frac{\tr(P\dot{H}_i)}{2\sigma^2}.
\end{equation*}
Apply the result in \eqref{eq:yphphpy}, we have
$$E(\I_Z(\kappa_i,\kappa_j))=0$$ and
$$
E(\I_A(\kappa_i,\kappa_j))=\frac{\tr(P\dot{H}_iP{H}_j)}{2}=I(\kappa_i,\kappa_j).
$$
\end{proof}

\section{Compute elements of averaged splitting information}
\label{sec:comp}

Compare $\I_A $ with $\I_O$, and $\I_F$ in Table \ref{tab:splitting}, in contrast with $I_{O}(\kappa_i,\kappa_j)$ which involves 4 matrix-matrix products, $\I_A(\kappa_i,\kappa_j)$ only involves a quadratic term
which can be evaluate by 4 matrix-vector multiplications and an inner product as in Algorithm~\ref{alg:IKK}.
\begin{algorithm}
\caption{Compute $\I_A(\kappa_i,\kappa_j)=\frac{y^TP\dot{H}_iP\dot{H}_jPy}{2\sigma_2}$}
\label{alg:IKK}
\begin{algorithmic}[1]
\State{ $\xi =Py$ }
\State{ $\eta_i =H_i \xi$; $\eta_j=H_j\xi$};
\State{ $\zeta= P\eta_j $}
\State{ $\I_A(\kappa_i ,\kappa_j)=\frac{\eta_i^T \xi}{2\sigma^2}$}
\end{algorithmic}
\end{algorithm}
One might think that $Py$ is complicated because of its representation in \eqref{eq:p}, whereas it turns
out that $Py$ have a very simple representation. We introduce the following lemma.
\begin{lemma}
Let $H=R+ZGZ^T$, then
$$
H^{-1}  =R^{-1} -R^{-1}Z(Z^TR^{-1}Z+G^{-1})^{-1} Z^TR^{-1}.
$$
\label{lem:H-1}
\end{lemma}
\begin{proof}
Using Fact~\ref{fact:1}.
\end{proof}
\begin{lemma}
The inverse of the matrix $C$ in \eqref{eq:mme} is given by
\[
C^{-1} =
\begin{pmatrix}
C_{XX} & C_{XZ} \\
C_{ZX} & C_{ZZ}
\end{pmatrix}^{-1}=
\begin{pmatrix}
C^{XX}  & C^{XZ} \\
C^{ZX} & C^{ZZ} \\
\end{pmatrix}
\]
where
\begin{align}
C^{XX} & =(X^TH^{-1}X)^{-1},\\
C^{XZ} & =-C^{XX}X^TR^{-1}ZC_{ZZ}^{-1}, \\
C^{ZX} &= -C_{ZZ}^{-1} Z^TR^{-1}XC^{XX}, \\
C^{ZZ} & =C_{ZZ}^{-1}+C_ZZ^{-1}Z^TR^{-1}XC^{XX}X^TR^{-1} Z^TC_{ZZ}^{-1}.
\end{align}
\end{lemma}
\begin{proof}
We only proof the formula
According to Fact~\ref{fact:S},
\begin{align*}
C^{XX} & =((X^TR^{-1}X)^{-1}-(X^TR^{-1}Z) C_{ZZ}^{-1}(Z^TR^{-1}X))^{-1} \\
 & = (X^T \underbrace{(R^{-1}-R^{-1}Z(Z^TR^{-1}Z+G^{-1})^{-1}Z^TR^{-1})}_{H^{-1}} X)^{-1} \\
 & =(X^TH^{-1}X)^{-1}.
\end{align*}
\end{proof}

We shall prove the following results
\begin{theorem}
Let $P$ be defined in \eqref{eq:p}, $\hat{\tau}$ and $\tilde{u}$ be the solution to \eqref{eq:mme}, and $e$ be the residual $e=y-X\hat{\tau}-Z\tilde{u}$, then
$Py=R^{-1} e$, and
\begin{align}
P&=H^{-1} -H^{-1}X(X^TH^{-1}X)^{-1}X^TH^{-1} \\
 & =R^{-1} -R^{-1}WC^{-1}W^TR^{-1} \label{eq:P2}
\end{align}
where $W=[X,Z]$ is the design matrix for the fixed and random effects.
\end{theorem}
\begin{proof} Suppose \eqref{eq:P2} hold, then
\begin{align}
Py &=R^{-1}y-R^{-1}W \underbrace{C^{-1}W^TR^{-1} y}_{(\hat{\tau}^T, \tilde{u}^T)^T} \\
  & =R^{-1}(y -X\hat{\tau} -Z\tilde{u}) =R^{-1}e
\end{align}
\begin{align*}
&  R^{-1} -R^{-1}WC^{-1}W^TR^{-1} \\
=& R^{-1}
-R^{-1}(X, Z) \begin{pmatrix} C^{XX} &  C^{XZ}  \\ C^{ZX} & C^{ZZ}
\end{pmatrix} \begin{pmatrix} X^T \\Z^T \end{pmatrix} R^{-1}. \\
= & R^{-1} -R^{-1}\{ XC^{XX}X^T -XC^{XZ}Z -ZC^{ZX} + ZC_{ZZ}^{-1}Z \\
& + Z(C_ZZ^{-1}Z^TR^{-1}XC^{XX}X^TR^{-1} Z^TC_{ZZ}^{-1})Z^T \}R^{-1} \\
 = &\underbrace{R^{-1} -R^{-1}ZC_{ZZ}^{-1}Z^TR^{-1}}_{H^{-1}}  \\
 & -(R^{-1} -R^{-1}ZC_{ZZ}^{-1}Z^TR^{-1}) XC^{XX}X^{T}H^{-1} \\
= & H^{-1} -H^{-1}X(X^TH^{-1}X)^{-1}X^TH^{-1}
\end{align*}
\end{proof}
From above results, we find out that evaluating the matrix vector $Py$ is equivalent the solve the linear system \eqref{eq:mme}, and then evaluate the weighted residual $R^{-1}e$. Notice that the matrix $P \in \mathbb{R}^{n\times n}$. On contrast, $C \in \mathbb{R}^{(p+b)\times (p+b)}$ where $p+b$ is the number of fixed effects and random effects. This number $p+b$ is much smaller than the number of observations $n$. In each nonlinear iterations, the matrix $C$  can be pre-factorized for evaluating $P\eta_i$.

\section{Discussion}

From above discussion, we know that he Fisher information has a simper form than the observed information and describes the essential information on the unknown parameters.
Therefore, Fisher information matrix is preferred not only in analyzing the asymptotic behavior
of maximum likelihood estimates \cite{EF76,Z97,Z98} but also in finding the variance of an
estimator and in Bayesian inference \cite{MN04}. In particular, if the Fisher information
matrix is used in the process of a maximum (log-)likelihood method,  which is widely used in
machine learning. Besides the traditional application fields like genetical theory of natural
selection and breeding \cite{F30},  many other fields including theoretical physics have
introduce the Fisher information matrix theory \cite{JJK04,PRE11,V07,HS14}.

The aim of information splitting is to remove computationally expensive and
negligible terms so that a much simper approximated information matrix is obtained. Such a splitting keeps
the essential information and can be used as a good approximation to the observed information matrix which is required for a derivative Newton method.
These formulas are much simpler than that used in genetics wide-association \cite{lip11,lis12,zhang10,zhou12},
and make derivative Newton method applicable for
large data sets which involve many thousands fixed and random effects.

\section*{Acknowledgement}
The first author would like to thank Dr Sue Welham at VSNi Ltd and
Professor Robin Thompson at Rothamsted Research and for introducing
him the restricted maximum likelihood method. The project is supported
by the Natural Science Fund of China (No.11501044)
and the Laboratory of Computational Physics.

\bibliographystyle{IEEEtran}
\bibliography{Reportv2}

\begin{thebibliography}{10}
\providecommand{\url}[1]{#1}
\csname url@samestyle\endcsname
\providecommand{\newblock}{\relax}
\providecommand{\bibinfo}[2]{#2}
\providecommand{\BIBentrySTDinterwordspacing}{\spaceskip=0pt\relax}
\providecommand{\BIBentryALTinterwordstretchfactor}{4}
\providecommand{\BIBentryALTinterwordspacing}{\spaceskip=\fontdimen2\font plus
\BIBentryALTinterwordstretchfactor\fontdimen3\font minus
  \fontdimen4\font\relax}
\providecommand{\BIBforeignlanguage}[2]{{%
\expandafter\ifx\csname l@#1\endcsname\relax
\typeout{** WARNING: IEEEtran.bst: No hyphenation pattern has been}%
\typeout{** loaded for the language `#1'. Using the pattern for}%
\typeout{** the default language instead.}%
\else
\language=\csname l@#1\endcsname
\fi
#2}}
\providecommand{\BIBdecl}{\relax}
\BIBdecl

\bibitem{masud13}
Y.~Masuda, T.~Baba, and M.~Suzuki, ``Application of supernodal sparse
  factorization and inversion to the estimation of (co) variance components by
  residual maximum likelihood,'' \emph{Journal of Animal Breeding and
  Genetics}, 2013.

\bibitem{Cell}
B.~M.~e. Bolker, ``Generalized linear mixed models:a practical guide for
  ecology and evolution,'' \emph{Trends in Ecology and Evolution}, vol.~24,
  no.~3, 2008.

\bibitem{lip11}
C.~Lippert, J.~Listgarten, Y.~Liu, C.~Kadie, R.~Davidson, and D.~Heckerman,
  ``Fast linear mixed models for genome-wide association studies,''
  \emph{Nature Methods}, vol.~8, no.~10, pp. 833--835, 2011.

\bibitem{lis12}
J.~Listgarten, C.~Lippert, C.~M. Kadie, R.~Davidson, E.~Eskin, and
  D.~Heckerman, ``Improved linear mixed models for genome-wide association
  studies,'' \emph{Nature methods}, vol.~9, no.~6, pp. 525--526, 2012.

\bibitem{zhang10}
Z.~Zhang, E.~Ersoz, C.-Q. Lai, R.~Todhunter, H.~K. Tiwari, M.~Gore,
  P.~Bradbury, J.~Yu, D.~Arnett, J.~Ordovas \emph{et~al.}, ``Mixed linear model
  approach adapted for genome-wide association studies,'' \emph{Nature
  genetics}, vol.~42, no.~4, pp. 355--360, 2010.

\bibitem{zhou12}
X.~Zhou and M.~Stephens, ``Genome-wide efficient mixed-model analysis for
  association studies,'' \emph{Nature genetics}, vol.~44, no.~7, pp. 821--824,
  2012.

\bibitem{Rao}
C.~E. McCulloch and S.~R. Searle, \emph{Generalized, linear, and mixed models},
  ser. Wiley Series in Probability and Statistics: Texts, References, and
  Pocketbooks Section.\hskip 1em plus 0.5em minus 0.4em\relax New York:
  Wiley-Interscience [John Wiley \& Sons], 2001.

\bibitem{V90}
A.~Verbyla, ``A conditional derivation of residual maximum likelihood,''
  \emph{Australian Journal of Statistics}, vol.~32, pp. 227--230, 1990.

\bibitem{Searle06}
S.~R. Searle, G.~Casella, and C.~E. McCulloch, \emph{Variance components}, ser.
  Wiley Series in Probability and Statistics.\hskip 1em plus 0.5em minus
  0.4em\relax Hoboken, NJ: Wiley-Interscience [John Wiley \& Sons], 2006,
  reprint of the 1992 original, Wiley-Interscience Paperback Series.

\bibitem{jenn76}
R.~Jennrich and P.~Sampson, ``Newton-raphson and related algorithms for maximum
  likelihood variance component estimation,'' \emph{Technometrics}, vol.~18,
  no.~1, pp. 11--17, 1976.

\bibitem{L87}
N.~Longford, ``A fast scoring algorithm for maximum likelihood estimation in
  unbalanced mixed models with nested random effects,'' \emph{Biometrika},
  vol.~74, no.~4, pp. 817--827, 1987.

\bibitem{ZGL16}
S.~Zhu, T.~Gu, and X.~Liu, ``Information matrix splitting,'' \emph{arXiv},
  2016, arXiv:1605.07646v1.

\bibitem{WZW13}
S.~Sue~Welham, S.~Zhu, and A.~J. Wathen, ``Big data, fast models: faster
  calculation of models from high-throughput biological data sets,'' Smith
  Industry Mathematics Institute, The University of Oxford, Oxford, Knowledge
  Transfer Project Reprot IP12-009, November 2013.

\bibitem{PT71}
H.~D. Patterson and R.~Thompson, ``Recovery of inter-block information when
  block sizes are unequal,'' \emph{Biometrika}, vol.~58, pp. 545--554, 1971.

\bibitem{all04}
L.~Wasserman, \emph{All of statistics: a concise course in statistical
  inference}.\hskip 1em plus 0.5em minus 0.4em\relax Springer, 2004.

\bibitem{Har97}
D.~A. Harville, \emph{Matrix Algebra From A Statistician's Perspective}.\hskip
  1em plus 0.5em minus 0.4em\relax Springer, 1997.

\bibitem{GTC95}
\BIBentryALTinterwordspacing
A.~R. Gilmour, R.~Thompson, and B.~R. Cullis,
  ``\BIBforeignlanguage{English}{Average information {REML}: An efficient
  algorithm for variance parameter estimation in linear mixed models},''
  \emph{\BIBforeignlanguage{English}{Biometrics}}, vol.~51, no.~4, pp. pp.
  1440--1450, 1995. [Online]. Available:
  \url{http://www.jstor.org/stable/2533274}
\BIBentrySTDinterwordspacing

\bibitem{Varga99}
R.~Varga, \emph{Matrix Iterative Analysis}, 2nd~ed.\hskip 1em plus 0.5em minus
  0.4em\relax Springer-Verlag.

\bibitem{EF76}
F.~Edgeworth and R.~Fisher, ``On the efficiency of maximum likelihood
  estimation,'' \emph{Annals of Statistics}, vol.~4, no.~3, pp. 501--514, 1976.

\bibitem{Z97}
R.~Zamir, ``A necessary and sufficient condition for equality in the matrix
  {F}isher information inequality,'' Tel Aviv University, Tech. Rep., 1997.

\bibitem{Z98}
------, ``A proof of the {F}isher information inequality via a data processing
  argument,'' \emph{IEEE Transaction On Information Theory}, vol.~44, pp.
  1246--1250, 2008.

\bibitem{MN04}
J.~Myung and D.~Navarro, ``Information matrix,'' in \emph{Encyclopedia of
  Behavioral Stastics}, B.~Everitt and D.~Howel, Eds., 2004.

\bibitem{F30}
R.~Fisher, \emph{The Genetical Theory of Natural Section}.\hskip 1em plus 0.5em
  minus 0.4em\relax Oxford:Clarendon Press, 1930.

\bibitem{JJK04}
W.~Janke, D.~Johnston, and R.~Kenna, ``Information geometry and phase
  transitions,'' \emph{Physica A}, vol. 336, no. 1-2, p. 181, 2004.

\bibitem{PRE11}
M.~Prokopenko, J.~T. Lizier, J.~T. Lizier, O.~Obst, and X.~R. Wang, ``Relating
  {F}isher information to order parameters,'' \emph{Physical Review E},
  vol.~84, no.~4, 2011.

\bibitem{V07}
M.~Vallisneri, ``A user manual for the fisher information matrix,'' California
  Institute of Technology, Jet Propulsion Laboratory, Tech. Rep., 2007.

\bibitem{HS14}
A.~Heavens, M.~Seikel, B.~Nord, M.~Aich, Y.~Bouffanais, B.~Bassett, and
  M.~Hobson, ``Generalised {F}isher matrices,'' \emph{Mon. Not. R. Astron.
  Soc.}, arXiv:1404.2854v2.

\end{thebibliography}

\appendix

\begin{fact}
Let $A \in \mathbb{R}^{n \times n} $, $B \in \mathbb{R}^{n\times m}$,
$C \in \mathbb{R}^{ m\times m}$, and $D \in \mathbb{R}^{m\times n}$. If $A$ is and $C$ is nonsingualr
then
$$
(A+BCD)^{-1}= A^{-1} -A^{-1}B(C^{-1}+DA^{-1}B)DA^{-1}.
$$
\label{fact:1}
\end{fact}
\begin{fact}
Let $A, B \in \mathbb{R}^{n\times n}$,then we have
\begin{align*}
(A+B)^{-1}A  & =I-(A+B)^{-1}B \\
A(A+B)^{-1}= I-B(A+B)^{-1}
\end{align*}
\end{fact}
\begin{proof}Notice the identity
$$ (A+B)^{-1}(A+B)=I =(A+B)(A+B)^{-1}. $$
\end{proof}
\begin{fact}
Suppose $S=(A-BC^{-1}B^T)^{-1}$ exist, then
\begin{equation*}
\begin{pmatrix}
A  & B  \\ B^T  & C
\end{pmatrix}^{-1}
=
\begin{pmatrix}
S  & -S BC^{-1} \\
-C^{-1}B^TS & C^{-1}B^TSBC^{-1}+C^{-1}
\end{pmatrix}.
\end{equation*}
\label{fact:S}
\end{fact}

\end{document}